\documentclass[11pt, a4paper]{article}

\usepackage[french, english]{babel}
\usepackage[utf8]{inputenc}
\usepackage[TS1,T1]{fontenc}

\usepackage{amsmath,amssymb,amsthm,amsfonts}    
\usepackage[T1]{fontenc}
\usepackage{appendix}                           
\usepackage{color}                              
\usepackage{enumerate}
\usepackage{enumitem}                           
\setlist{nolistsep}

\usepackage[a4paper,tmargin=3cm,bmargin=3cm,rmargin=2.5cm,lmargin=2.5cm]{geometry}

\usepackage{fancyhdr}
\fancypagestyle{plain}{%
\pagestyle{fancy}

\lhead{\nouppercase{\rightorleftmark}}

\fancyhead[C]{} 
\fancyhead[L]{}
\fancyhead[R]{}

\fancyfoot[C]{\textbf{\thepage}} 
\fancyfoot[L]{}
\fancyfoot[R]{}
}
\usepackage{float}

\usepackage[bookmarksnumbered,pdftex]{hyperref}
\hypersetup{
    bookmarks=true,         
    unicode=false,          
    pdftoolbar=true,        
    pdfmenubar=true,        
    pdffitwindow=false,     
    pdfstartview={FitH},    
    pdftitle={Examples of infinite direct sums of spectral triples},
    pdfauthor={Kevin Falk},
    pdfsubject={},
    pdfcreator={Kevin Falk},
    pdfproducer={Kevin Falk},
    pdfkeywords={Noncommutative geometry, spectral triples, self-similar sets}
    pdfnewwindow=false,     
    colorlinks=true,        
    linkcolor=blue,         
    citecolor=magenta,      
    filecolor=magenta,      
    urlcolor=cyan           
}

\usepackage{nameref}
\usepackage{mathrsfs}                           
\usepackage{mathtools}                          %
\usepackage{minibox}                            
\usepackage[nottoc, notlof, notlot]{tocbibind}  %
\usepackage{slashed}                            
\usepackage{supertabular}                       
\usepackage{tabularx}                           
\usepackage{subcaption}                         

\usepackage{tikz}
\usetikzlibrary{matrix}
\usetikzlibrary{lindenmayersystems} 
\usetikzlibrary{decorations}

\usepackage{xcolor}
\usepackage[outline]{contour}
\contourlength{1.7pt}

\newcommand{\A}{\mathcal{A}}                                    
\newcommand{\abs}[1]{\vert  #1 \vert}                           
\newcommand{\B}{\mathbb{B}}                                     
\newcommand{\C}{\mathbb{C}}                                     
\newcommand{\commut}[2]{[\, #1,\,#2 \,] }                           
\newcommand{\Coo}{C^\infty}                                     
\newcommand{\cv}{=\vcentcolon }                                 
\newcommand{\DD}{\mathcal{D}}                                   
\renewcommand{\dim}{{\rm dim}}                                  
\newcommand{\dom}{{\rm dom}}                                    
\newcommand{\eps}{\varepsilon}                                  
\newcommand{\evalat}[1]{|_{#1}}                                 
\newcommand{\fat}[1]{\boldsymbol{#1}}                           
\newcommand{\fatR}[0]{\fat{\mathcal{R}}}                                  
\newcommand{\fatT}[0]{{\bf T}}                                  
\renewcommand{\H}{\mathcal{H}}                                  
\newcommand{\hideqed}{\renewcommand{\qed}{}}                    
\newcommand{\Ker}{{\rm Ker}}                                    
\newcommand{\noNull}{\backslash\{0\}}                           
\newcommand{\N}{\mathbb{N}}                                     
\newcommand{\norm}[1]{\lVert \, #1 \, \rVert}                           
\newcommand{\omegaj}{{\omega j}}				
\newcommand{\Pol}{{\rm Pol}}                                    
\newcommand{\R}{\mathbb{R}}                                     
\newcommand{\Ran}{{\rm Ran}}                                    
\newcommand{\scalp}[2]{\langle {\,#1 \, ,\,#2\,} \rangle}           
\newcommand{\Spec}{{\rm Spec}}                                  
\newcommand{\Tr}{{\rm Tr}}                                      
\newcommand{\vc}{\vcentcolon =}                                 

\newcommand{\appref}[1]{\hyperref[#1]{Appendix \ref{#1}}}                       
\newcommand{\defref}[1]{\hyperref[#1]{Definition \ref{#1}}}                     
\newcommand{\chapref}[1]{\hyperref[#1]{Chapter \ref{#1}}}                       
\newcommand{\condref}[1]{\hyperref[#1]{Condition \ref{#1}}}                     
\newcommand{\condreff}[2]{\hyperref[#1]{Condition \ref{#1}, \ref{#2}}}           
\newcommand{\corref}[1]{\hyperref[#1]{Corollary \ref{#1}}}                      
\newcommand{\exref}[1]{\hyperref[#1]{Example \ref{#1}}}                         
\newcommand{\figref}[1]{\hyperref[#1]{Figure \ref{#1}}}                         
\newcommand{\lemref}[1]{\hyperref[#1]{Lemma \ref{#1}}}                          
\newcommand{\propref}[1]{\hyperref[#1]{Proposition \ref{#1}}}                   
\newcommand{\propGTOref}[1]{(\hyperref[#1]{{\autoref{#1}}})}                    
\newcommand{\remref}[1]{\hyperref[#1]{Remark \ref{#1}}}                         
\newcommand{\remsref}[1]{\hyperref[#1]{Remarks \ref{#1}}}                       
\newcommand{\secref}[1]{\hyperref[#1]{Section \ref{#1}}}                        
\newcommand{\theoref}[1]{\hyperref[#1]{Theorem \ref{#1}}}                       

\DeclareMathSymbol{\Kappa}{\mathalpha}{operators}{"4B}

\newtheorem{theorem}{Theorem}[section]

\newtheorem{definition}[theorem]{Definition}
\newtheorem{lemma}[theorem]{Lemma}
\newtheorem{proposition}[theorem]{Proposition}
\newtheorem{example}[theorem]{Example}

\newtheorem{remark}[theorem]{Remark}

\newcounter{mnotecount}[section]
\renewcommand{\themnotecount}{\thesection.\arabic{mnotecount}}
\newcommand{\mnote}[1]
{\protect{\stepcounter{mnotecount}}$^{\mbox{\footnotesize $\bullet$\themnotecount}}$ \marginpar{\raggedright\tiny\em$\!\!\!\!\!\!\,\bullet$\themnotecount: #1} }

\pgfdeclaredecoration{recsierpinski2}{do}{%
    \state{do}[width=\pgfdecoratedinputsegmentlength, next state=do]{%
 	\pgfpathmoveto{\pgfpointpolar{0}{0}}%
 	\pgfpathlineto{\pgfpointpolar{0}{\pgfdecoratedinputsegmentlength}}%
 	\pgfpathmoveto{\pgfpointpolar{0}{0}}%
        \pgfpathmoveto{\pgfpointpolar{0}{\pgfdecoratedinputsegmentlength/2}}%
        \pgfpathlineto{\pgfpointpolar{-60}{\pgfdecoratedinputsegmentlength/2}}%
        \pgfpathlineto{\pgfpointpolar{-30}{\pgfdecoratedinputsegmentlength/2*1.732050}}%
        \pgfpathlineto{\pgfpointpolar{0}{\pgfdecoratedinputsegmentlength/2}}%
    }
}

\title{Examples of infinite direct sums of spectral triples}
\author{Kevin Falk}
\date{}

\begin{document}
\maketitle

  \begin{center}
    Centre de Physique Théorique,\\
    Aix Marseille Université \& Université de Toulon \&  CNRS UMR 7332,\\
    13288 Marseille, France
  \end{center}

  \begin{abstract}
    We study two ways of summing an infinite family of noncommutative spectral triples. First, we propose a definition of the integration of spectral triples and give an example using algebras of Toeplitz operators acting on weighted Bergman spaces over the unit ball of $\C^n$. Secondly, we construct a spectral triple associated to a general polygonal self-similar set in $\C$ using algebras of Toeplitz operators on Hardy spaces. In this case, we show that we can recover the Hausdorff dimension of the fractal set.
  \end{abstract}
  
  \emph{Keywords:} noncommutative geometry, spectral triples, Toeplitz operators, self-similar sets.

  \section{Introduction and motivation}
    
    The main idea of Connes's noncommutative geometry is to characterize the geometry of a space in the language of algebras \cite{C1994}. We know for instance that a compact Hausdorff space can be equivalently seen as the \emph{commutative} $C^*$-algebra of continuous functions living on it. By analogy, a \emph{noncommutative} algebra would correspond to a space of quantum nature: a \emph{noncommutative space}. More precisely, the algebraic description of a Riemannian manifold is based on the notion of \emph{unital spectral triple}, consisting of the data $(\A,\H,\DD)$, where $\A$ is an involutive unital $*$-algebra $\A$ faithfully represented  on a Hilbert space $\H$ via a representation $\pi$, and $\DD$ is a selfadjoint operator acting on $\H$ with compact resolvent and such that for any $a\in\A$, $\pi(a)$ maps $\dom(\DD)$ into itself, and $\commut{\DD}{\pi(a)}$ extends to a bounded operator on $\H$. When $\A$ is not unital, replace the compactness of the resolvent by the compactness of $\pi(a)(\DD - \lambda )^{-1}$ for any $a\in\A$ and $\lambda\notin \Spec(\DD)$: the induced triple is then called \emph{nonunital}.
    Among the various geometric entities which are encoded in the spectrum of $\DD$, we are interested in the so-called \emph{spectral dimension}, defined as the quantity 
    \begin{align*}
      d \vc \inf\{ s \in \R\,, \Tr \,\abs{\DD}^{-s} < +\infty\}\,.
    \end{align*}
    As easily checked, the direct sum of a finite number of spectral triples is again a spectral triple. We are interested here in \emph{integrations} of spectral triples which consist, roughly speaking, of the direct sum of an \emph{infinite} number of spectral triples. Such constructions have already been encountered in \cite{EFI2012}: the spectral triple related to the Berezin--Toeplitz quantization over a smoothly bounded stricly pseudoconvex domain of $\C^n$ can be viewed as the integration of an infinite family of spectral triples based on algebras generated by Toeplitz operators acting on weighted Bergman spaces.
    
    The first idea is the following: given a countable family of spectral triples $(\A_m,\H_m,\DD_m)_{m\in\N}$ (commutative or not), the corresponding infinite direct sum ``$\bigoplus_{m\in\N}(\A_m,\H_m,\DD_m)$'', might not be necessarily a spectral triple again. Indeed, as $m$ tends to infinity, the boundedness of the representations of $\A_m$, the boundedness of the commutator between $\A_m$ and the operators $\DD_m$, or the compactness of the resolvent of the direct sum of all operators $\DD_m$ is hard to control in general and the sum may fail to converge. In order to control the behaviour of the operators $\DD_m$, we multiply them by some coefficients $\alpha_m \in\R\noNull$.
    
    Surprisingly, a strong link exists between direct summations of spectral triples and fractal sets, but before describing the second approach, let us recall some previous results on the topic. Since the works of A. Connes \cite[Chapter 4, 3.$\eps$]{C1994}, we know that noncommutative geometry can detect the topology of fractal sets: it is shown that a commutative spectral triple involving the $C^*$-algebra of continuous functions over the Cantor set can be used to recover its Hausdorff dimension and the Hausdorff measure. Later on, D. Guido and T. Isola proposed a commutative spectral triple, also based on a discrete approximation of the fractal, and extend Connes' result to more general self-similar sets in $\R^n$ \cite[Chapter 7]{GI2001}, \cite{GI2003} (the existence of such spectral triples was already conjectured in M. Lapidus' paper \cite{L1994}). See also \cite{L1997} for a review of open problems and questions about the links between analysis and spectral geometry on fractal sets. 
    
    In the latter works, each spectral triple is directly built over the fractal set. The approach we follow in the present paper is a constructive one: decompose the considered fractal set as the union of an infinite number of subdomains and associate to each of them a spectral triple. The spectral triple over the whole fractal set is obtained after the direct summation of all these spectral triples. This construction has already been used in \cite{CIL2008,CGIS2014,LS2014} to recover the Hausdorff dimension and the metric on $p$-summable infinite trees and the Sierpinski gasket, and also in \cite{CIS2012} to study the Hausdorff dimension of the Sierpinski gasket (and pyramid), its metric and describe its K-homology group.
    
    For simplicity reasons, we restrict our study to self-similar sets $E$ of the plane $\C$ which can be expressed as
    \begin{align}
      \label{eqE}
      E = \overline{E_0 \cup \bigcup_{k=1}^NF_k(E_0) \cup \bigcup_{k,l=1}^NF_k\circ F_l(E_0) \cup \dots}\,,
    \end{align}
    where the overline means taking the closure, $E_0$ is a polygonal Jordan curve in the complex plane or the unit disk, and $(F_k)_{k=1,\dots,N}$ is a finite family of contracting similarities. 
    
    The paper is organized as follows.\\
    We present in \secref{SecAbstractIntegration} some sufficient conditions for the sum to be a spectral triple and we give an example of such integration using Toeplitz operators over the unit ball of $\C^n$.\\ 
    We show in \secref{secFractalIntegration} that is is possible to build a \emph{noncommutative} spectral triple over such sets, involving algebras of Toeplitz operators, and whose spectral dimension corresponds to the Hausdorff dimension of $E$.

  \section{Abstract integration of spectral triples}
    \label{SecAbstractIntegration}
    
    \subsection{Conditions of integrability}

      \begin{lemma}
        \label{lemDoplusEssSelfAdjoint}
        Let $(\H_m)_{m\in\N}$ be a family of Hilbert spaces, $(\DD_m)_{m\in\N}$ be a family of unbounded selfadjoint operators with corresponding dense domains $(\dom(\DD_m) \subset \H_m)_{m\in\N}$, and $(\alpha_m)_{m\in\N} \in (\R\noNull)^\N$. Let $\DD^\oplus \vc \bigoplus_{m\in\N} \alpha_m\DD_m$ with domain
        \begin{align*}
          \dom(\DD^\oplus) \vc \big\{\bigoplus_{m=0}^N v_m \in \H^\oplus\,, N\in\N\,, v_m\in\dom(\DD_m)\, \big\}\,.
        \end{align*}
        Then $\DD^\oplus$ is essentially selfadjoint, with selfadjoint extension $\overline{\DD^\oplus}$.
      \end{lemma}
      \begin{proof}
        Let $v^\oplus\vc \bigoplus_{m\in\N} v_m \in \H^\oplus$. For any $m\in\N$, the operator $\DD_m$ is densely defined so there is a sequence $(v_{mj})_{j\in\N}$ of elements in $\dom(\DD_m)$ converging to $v_m$ as $j\to\infty$. Thus for any fixed $(m,j)\in\N^2$, there is $M_{mj}\in\N$ such that $\norm{v_m - v_{m,M_{mj}+k}}_{\H_m}^2 < 2^{-j}$ for any $k\in\N$. Define for any $j\in\N$ the vector $w_j^{\oplus} \vc \bigoplus_{m=0}^j v_{m,M_{mj}} \in \dom(\DD^\oplus)$. For any $j\in\N$, $w_j^\oplus \in \dom(\DD^\oplus)$ and
        \begin{align*}
          \norm{v^\oplus - w_j^\oplus}^2_{\H^\oplus} = \sum_{m=0}^j \norm{v_m - v_{m,M_{mj}+k}}^2_{\H_m} + \sum_{m>j}\norm{v_m}^2_{\H_m} < j2^{-j} +  \sum_{m>j}\norm{v_m}^2_{\H_m} \underset{j\to +\infty}{\longrightarrow} 0\,.
        \end{align*}
        Thus for any $\eps>0$, there exists $N\in\N$ such that $\norm{v^\oplus-w_N^\oplus}_{\H^\oplus} < \eps$, which shows that $\DD^{\oplus}$ is densely defined.\\
        Using the same reasoning and the fact that for any $m\in\N$, $\Ran(\alpha_m\DD_m \pm i) = \H_m$ (since $\alpha_m\DD_m$ is selfadjoint), it can be shown that for any $v^\oplus\in \H^\oplus$ and $\eps>0$, there is $N\in\N$ and $w_N^\oplus\in\dom(\DD^\oplus)$ defined as above and such that $\norm{v^\oplus - (\DD^\oplus \pm i)w_N^\oplus}_{\H^\oplus}<\eps$, thus $\Ran(\DD^\oplus \pm i)$ is dense in $\H^\oplus$.\\
        The operator $\DD^\oplus$ is also symmetric since for any $v^\oplus\vc \sum_{m=0}^{N}v_m$ and $ v'^\oplus \vc \sum_{m=0}^{N'}v'_m$ in $\dom(\DD^\oplus)$,
        \begin{align*}
          \scalp{\DD^\oplus v^\oplus}{v'^\oplus}_{\H^{\oplus}} & = \!\!\!\sum_{m=0}^{\min(N,N')}\!\!\! \scalp{\alpha_m\DD_m v_m}{v'_m}_{\H_m}  = \!\!\!\sum_{m=0}^{\min(N,N')} \!\!\!\scalp{v_m}{\alpha_m\DD_m v'_m}_{\H_m}  = \scalp{v^\oplus}{\DD^\oplus v'^\oplus}_{\H^{\oplus}}\,,
        \end{align*}
        which shows that $\DD^\oplus$ is essentially selfadjoint (see \cite[Chapter VIII.2, Corollary p.257]{simon1980methods}).
      \end{proof}
      
      The following result establishes sufficient conditions on an infinite family of spectral triples together with a family of weights $(\alpha_m)_{m\in\N}\in(\R\noNull)^\N$ so that the corresponding weighted direct sum is a spectral triple. 
      
      \begin{proposition}
        \label{propSumST}
        Let $(\A_m, \H_m, \DD_m)_{m\in\N}$ be a family of (not necessarily unital) spectral triples, with corresponding representations $(\pi_m)_{m\in\N}$, and denote $\norm{.}_m$ the norm on $\H_m$.\\
        Let $(\alpha_m)_{m\in\N}$ be a sequence of non-zero real numbers such that
        \begin{align}
          \label{condInv}
          \norm{(1+\alpha_m^{2}\DD_m^2)^{-1/2}}_m \,\underset{m\to +\infty}{\longrightarrow}\,0\,.
        \end{align}
        Define the following objects:
        \begin{itemize}
          \item $\H^\oplus \vc \bigoplus_{m\in\N} \H_m$,
          \item $\DD^\oplus \vc \bigoplus_{m\in\N} \alpha_m \,\DD_m$ and $\overline{\DD^\oplus}$ as above, both acting on $\H^\oplus$,
          \item $\displaystyle\begin{aligned}[t]
                    \!\A^\oplus \vc \big\{
					(a_m)_{m\in\N}\in\prod_{m\in\N}\A_m: \,&  \sup_{m\in\N} \norm{\pi_m(a_m)}_m < +\infty\,, \text{ and }\\
					& \sup_{m\in\N}\, \norm{\,\commut{\alpha_m\DD_m}{\pi_m(a_m)}\,}_m<+\infty
				    \big\},
			      \end{aligned}$
          \item $\pi^\oplus(a^\oplus) \vc \bigoplus_{m\in\N} \pi_m(a_m)$, for $a^\oplus\in \A^\oplus$.
        \end{itemize}
        Then $(\A^{\oplus}, \H^{\oplus}, \overline{\DD^{\oplus}})$ is a (not necessarily unital) spectral triple.
      \end{proposition}
      \begin{proof}
	
	For two elements $a^\oplus=(a_m)_{m\in\N}$ and $b^\oplus=(b_m)_{m\in\N}$ in $\A^\oplus$, we have:
        \begin{align*}
          & \sup_{m\in\N} \norm{\pi_m(a_mb_m)}_m \leq \sup_{m\in\N} \norm{\pi_m(a_m)}_m\, \sup_{m\in\N} \norm{\pi_m(b_m)}_m < +\infty \,, \quad \text { and }\\
          & \sup_{m\in\N} \norm{\commut{\alpha_m\DD_m}{\pi_m(a_mb_m)}}_m \leq \sup_{m\in\N} \norm{\pi_m(a_m)}_m\, \sup_{m\in\N}\norm{\commut{\alpha_m\DD_m}{\pi_m(b_m)}}_m  \\
          &\hspace{5.1cm}+\sup_{m\in\N}\norm{\commut{\alpha_m\DD_m}{\pi_m(a_m)}}_m \,\sup_{m\in\N}\norm{\pi_m(b_m)}_m < +\infty\,,
        \end{align*}
        hence $\A^\oplus$ is an algebra with involution $*:a^\oplus=(a_m)_{m\in\N} \mapsto (a^\oplus)^*\vc (a_m^*)_{m\in\N}$.\\
	For $a^\oplus \in\A^\oplus$, we have
        \begin{align*}
          \pi^\oplus(a^\oplus)\,\big(1+(\DD^\oplus)^2\big)^{-1/2}= \bigoplus_{m\in\N} \pi_m(a_m)\,(1+\alpha_m^2\,\DD_m^{2})^{-1/2}.
        \end{align*}
        For any $m\in\N$, the summand $\pi_m(a_m)\,(1+\alpha_m^2\,\DD_m^{2})^{-1/2}$ is compact. 
        From \eqref{condInv} and the fact that $\pi^\oplus$ is a bounded representation, $\norm{\pi_m(a_m)\,(1+\alpha_m^2\,\DD_m^{2})^{-1/2}}_m$ tends to 0 as $m\to+\infty$. 
        As a consequence, $\pi^\oplus(a^\oplus)\,\big(1+(\DD^\oplus)^2\big)^{-1/2}$ is compact.\\
        From \lemref{lemDoplusEssSelfAdjoint}, $\DD^\oplus$ is essentially selfadjoint with selfadjoint extension $\overline{\DD^\oplus}$.\\
        For $a^\oplus\in\A^\oplus$ and $v_N^\oplus \vc \bigoplus_{m=0}^Nv_{N,m} \in \dom(\DD^\oplus)$, for some $N\in\N$, we have 
        \begin{align*}
	  \pi^\oplus(a^\oplus)v_N^\oplus = \bigoplus_{m=0}^N \pi_m(a_m)v_{N,m}
	\end{align*}
	and each summand on the right-hand side belongs to $\dom(\DD_m)$ since $(\A_m,\H_m,\DD_m)$ is a spectral triple for any $m\in\N$. Thus $\pi^\oplus(a^\oplus)$ maps $\dom(\DD^\oplus)$ into itself for any $a^\oplus\in\A^\oplus$. \\
        Moreover, for any $a^\oplus\in\A^\oplus$ and $v_N^\oplus \vc \bigoplus_{m=0}^Nv_{N,m} \in \dom(\DD^\oplus)$ of norm 1, we have
        \begin{align*}
          \norm{\commut{\DD^\oplus}{\pi^\oplus(a^\oplus)}v_N^\oplus} = \sup_{m=0,\dots,N} \norm{\commut{\alpha_m\DD_m}{\pi_m(a_m)}v_{N,m}} \leq  \sup_{m\in\N}\, \norm{\,\commut{\alpha_m\DD_m}{\pi_m(a_m)}\,}_m<+\infty\,,
        \end{align*}
        so $\commut{\DD^\oplus}{\pi^\oplus(a^\oplus)}$ is bounded on $\dom(\DD^\oplus)$. Moreover, since 
        \begin{align*}
          \overline{\overline{\DD^\oplus} \evalat{\dom(\DD^\oplus)}} = \overline{\DD^\oplus\evalat{\dom(\DD^\oplus)}} = \overline{\DD^\oplus}\,,
        \end{align*}
	then $\dom(\DD^\oplus)$ is a core for $\overline{\DD^\oplus}$. Using \cite[Proposition A.1]{P2014}, we conclude that for any $a^\oplus\in\A^\oplus$,
        \begin{align*}
          \pi^\oplus(a^\oplus)\big(\dom(\overline{\DD^\oplus}) \big)\subset \dom(\overline{\DD^\oplus})
        \end{align*}
	and $\commut{\overline{\DD^\oplus}}{\pi^\oplus(a^\oplus)}$ extends to a bounded operator on $\H^\oplus$.
      \end{proof}

      \begin{definition}
	The spectral triple $(\A^{\oplus}, \H^{\oplus}, \overline{\DD^{\oplus}})$ as above is called the integration of the five-tuple $(\A_m, \H_m, \DD_m, \pi_m, \alpha_m)_{m\in\N}$, where  $(\A_m, \H_m, \DD_m)$ is a spectral triple for any $m\in\N$, with corresponding representations $\pi_m$ and weights $\alpha_m$ in $\R\noNull$.
      \end{definition}
      As a consequence of \eqref{condInv}, the sequence $(\DD_m)_{m\in\N}$ is such that $\sum_{m\in\N} \dim \,( \Ker \,\DD_m)<\infty$. In particular, if we take the same $\DD_m=\DD_0$ at each level $m\in\N$, the latter must be invertible.
      
      The two conditions in the definition of $\A^\oplus$ correspond to the boundedness of both the representation $\pi^{\oplus}$ and the commutator $\commut{\DD^\oplus}{\pi^\oplus(\A^\oplus)}$ for the norm $\norm{.}^{\oplus} \vc \sup_{m\in\N} \norm{.}_m$ on $\pi^\oplus(\A^\oplus)$.\\
      The parameter $(\alpha_m)_{m\in\N}$ has been introduced in order to control the behaviour of the sequence $(\DD_m)_{m\in\N}$ as $m$ tends to infinity. This can be avoided by putting some constraints directly on the operators $\DD_m$, but this restricts the set of summable families of spectral triples. For instance, when $\DD^\oplus \vc \bigoplus_{m\in\N} \DD_0$, with $\DD_0$ invertible, then the resolvent of $\DD^\oplus$ is not compact.
        
      We make use of the following notations for the rest of the document. For a multiindex $\alpha\in\N^n$ and $z\in\C^n$, denote $z^\alpha \vc z_1^{\alpha_1}z_2^{\alpha_2}\dots z_n^{\alpha_n}$ and $\abs{\alpha} \vc \alpha_1+\dots+\abs{\alpha_n}$. For any set $X\subset\C^n$, let $\Pol(X)$ be the set of polynomial functions in $z$ and $\bar{z}$ over $X$. We denote $\B^n \vc \{ z\in\C^n\,, \, \abs{z}< 1\}$ the unit open ball of $\C^n$, $\overline{\B^n}$ its closure, and simply $\B \vc \B^1$ the unit open disk in $\C$. The corresponding boundaries are denoted respectively $\partial\B^n$ and $\partial\B$.
      
    \subsection{An example of integration over the unit ball}
      \label{secExampleIntegration}

      We choose for $\B^n$ the following defining function (i.e. a smooth function $r$ over $\overline{\B^n}$ such that $r\evalat{\B^n}<0$, $r\evalat{\partial\B^n} = 0$ and $dr\evalat{\partial\B^n} \neq 0$) and a weight on $\B^n$:
      \begin{align}
        \label{eqDefFunctionAndWeight}
        r(z) \vc \abs{z}^2 - 1 \,, \quad \text{ and } \quad w_{m}(z) \vc (-r(z))^{m}\,, \quad z\in \overline{\B^n}\,, \quad m\in(-1,+\infty)\,.
      \end{align}
      The \emph{weighted Bergman space} over $\B^n$ with weight $w_{m}$ is
      \begin{align*}
        A^2_{m}(\B^n) \vc \{ \phi\in L^2(\B^n, w_{m} d\mu) \,, \text{ $\phi$ holomorphic in $\B^n$} \}\,,
      \end{align*}
      where $d\mu$ is the usual normalized Lebesgue measure over $\B^n$. Denote $\boldsymbol{\Pi}_{m}$ the orthogonal projection from $L^2(\B^n)$ onto $A^2_{m}(\B^n)$. The \emph{Toeplitz operator} $\fatT^{(m)}_f : A^2_{m}(\B^n) \to A^2_{m}(\B^n)$ associated to the function $f\in \Coo(\overline{\B^n})$ is defined as
      \begin{align*}
        \fatT^{(m)}_f : \phi \mapsto \boldsymbol{\Pi}_{m} (f\phi) \,.
      \end{align*}
      In particular, Toeplitz operators enjoy the following properties: 
      \begin{align}
	\label{eqPropToeplitz}
        f\mapsto \fatT^{(m)}_f \text{ is linear}\,, \quad \norm{\fatT_f^{(m)}} \leq \norm{f}_\infty \,, \quad \text{ and } \quad (\fatT^{(m)}_f)^* = \mathbf{T}^{(m)}_{\bar f}\,.
      \end{align}
      Since in general the product of two Toeplitz operators is not a Toeplitz operator anymore, we will consider the $*$-algebra \emph{generated} by the Toeplitz operators, the involution being the Hilbert space adjoint operation \eqref{eqPropToeplitz}.
      
      The following result is a corollary of \cite[Proposition 5.4]{EFI2012}:
      \begin{proposition}
        \label{propTripleBergman}
        For any real number $m>-1$, let $\A_{m}$ be the algebra generated by the Toeplitz operators $\mathbf{T}^{(m)}_f$, $f\in\Coo(\overline{\B^n})$, with the identity representation on $\H_{m} \vc A^2_{m}(\B^n)$, and also $\DD_{m} \vc (\mathbf{T}_{-r}^{(m)})^{-1}$.\\
        Then $(\A_{m}\,,\H_{m}\,,\DD_{m})$ is a spectral triple of spectral dimension $n=\dim_{\C}\,\B^n$.
      \end{proposition}
      
      In order to get a family of spectral triples, we make $m$ vary in the set of integers, and for the rest of this section $m \in\N$. Let us present a preliminary result which establishes the dependence on $m$ of the commutator between a Toeplitz operator with polynomial symbol and the previous operator $\DD_m \vc (\fatT_{-r}^{(m)})^{-1}$. Denote the operators $\fat{\mathcal{R}} \vc \sum_{j=1}^n \fat{\mathcal{R}}_j$ and $\overline{\fat{\mathcal{R}}} \vc \sum_{j=1}^n \overline{\fat{\mathcal{R}}}_j$ with $\fat{\mathcal{R}}_j \vc z_j\partial_{z_j}$ and $\overline{\fat{\mathcal{R}}}_j \vc \bar{z}_j\,\partial_{\bar{z}_j}$, acting on $\Coo(\overline{\B^n})$. 
      
      \begin{proposition}
        \label{propCommutatorInAlg}
        For any polynomial function $p(z) = \!\!\!\!\!\!\underset{\abs{\alpha}\leq d,\abs{\beta}\leq d'}{\sum} p_{\alpha\beta}\, z^\alpha\,\bar{z}^\beta \in\Pol(\mathbb{B}^n)$, we have
        \begin{align*}
          \commut{(\mathbf{T}^{(m)}_{-r})^{-1}}{\mathbf{T}^{(m)}_p} = \tfrac{1}{m+1} \, \mathbf{T}^{(m)}_{(\fat{\mathcal{R}}-\overline{\fat{\mathcal{R}}})\,p}\,, \quad \text{  on $A^2_m(\mathbb{B}^n)$}\,.
        \end{align*}
      \end{proposition}
      
      \begin{proof}
        We denote briefly $\mathbf{T}_{p} = \mathbf{T}^{(m)}_{p}$. An orthonormal basis of $A^2_m(\mathbb{B}^n)$ is given by (see \cite[(2.9)]{Z2005})
        \begin{align}
          \label{eqONBBergmanBn}
          u_{m,\alpha}(z)  \vc \big(\tfrac{(\abs{\alpha}+m+n)!}{(m+n)!\,\alpha!} \big)^{1/2}\, z^\alpha\,.
        \end{align}
        Using the shift operators $\fat{S}_j: u_{m,\alpha} \mapsto u_{m,\alpha+1_j}$, with $\alpha+1_j \vc (\alpha_1, \dots, \alpha_j +1, \dots, \alpha_n)$ and $j=1,\dots,n$, we have the relations
        \begin{align}
          \label{eqRelationsShift}
          \mathbf{T}_{z_j} &= \fat{S}_j\,(\tfrac{\fat{\mathcal{R}}_j+1}{\fat{\mathcal{R}}+m+n+1})^{1/2}\, , \quad \commut{\fat{\mathcal{R}}_j}{\fat{S}_k} = \delta(j=k)\,\fat{S}_j \,, \quad \fat{S}_j^* \fat{S}_j=1\,, \quad \text{ for } j=1,\dots,n, \quad \text{ and }\\
          \mathbf{T}_{-r}^{-1} & = (1-\sum_{j=1}^n\mathbf{T}_{\abs{z_j}^2})^{-1} = (1-\sum_{j=1}^n (\mathbf{T}_{z_j})^*\mathbf{T}_{z_j})^{-1} = (1-\sum_{j=1}^n\tfrac{\fat{\mathcal{R}}_j+1}{\fat{\mathcal{R}}+m+n+1})^{-1} \nonumber \\
          &= \tfrac{1}{m+1}(\fat{\mathcal{R}}+m+n+1)\,. \nonumber
        \end{align}
        Hence we get
        \begin{align*}
          \commut{\mathbf{T}_{-r}^{-1}}{\mathbf{T}_{z_j}} & = \tfrac{1}{m+1}\, \big(\,(\fat{\mathcal{R}}+m+n+1)\, \fat{S}_j\,(\tfrac{\fat{\mathcal{R}}_j+1}{\fat{\mathcal{R}}+m+n+1})^{1/2} - \fat{S}_j\,(\tfrac{\fat{\mathcal{R}}_j+1}{\fat{\mathcal{R}}+m+n+1})^{1/2} (\fat{\mathcal{R}}+m+n+1) \,\big) \\
          & = \tfrac{1}{m+1} \, \fat{S}_j\,(\tfrac{\fat{\mathcal{R}}_j+1}{\fat{\mathcal{R}}+m+n+1})^{1/2} \,\big( \fat{\mathcal{R}}+m+n+2 - (\fat{\mathcal{R}}+m+n+1) \big)\\
          &= \tfrac{1}{m+1}\,\mathbf{T}_{z_j}\,.
        \end{align*}
        From this last equality and the fact that $\commut{\mathbf{T}_{z_j}}{\mathbf{T}_{z_k}}=0$, for any $j,k=1,\dots,n$, we get by iteration of the formula $\commut{A}{BC} = B\commut{A}{C} + \commut{A}{B}C$
        \begin{align*}
          & \commut{\mathbf{T}_{-r}^{-1}}{\prod_{j=1}^n \mathbf{T}_{z_j}^{\alpha_j}} = \tfrac{\abs{\alpha}}{m+1}\,\prod_{j=1}^n \mathbf{T}_{z_j}^{\alpha_j} \quad \text{ and } \quad 
          \commut{\mathbf{T}_{-r}^{-1}}{\prod_{j=1}^n (\mathbf{T}^*_{z_j})^{\beta_j}} = -\tfrac{\abs{\beta}}{m+1}\,\prod_{j=1}^n (\mathbf{T}^*_{z_j})^{\beta_j}\,, \quad \alpha,\beta\in\N^n\,.
        \end{align*}
        Hence, the relation $\mathbf{T}_{z^\alpha\bar{z}^\beta} = \big(\prod_{j=1}^n (\mathbf{T}^*_{z_j})^{\beta_j} \big)\big(\prod_{j=1}^n \mathbf{T}_{z_j}^{\alpha_j}\big)$ yields to
        \begin{align*}
          \commut{\mathbf{T}_{-r}^{-1}}{\mathbf{T}_p} & = \!\!\!\!\sum_{\abs{\alpha}\leq d, \abs{\beta}\leq d'}  \!\!\!\!\!\!\!\!\,\commut{\mathbf{T}_{-r}^{-1}}{\mathbf{T}_{z^\alpha\bar{z}^\beta}} = \!\!\!\!\sum_{\abs{\alpha}\leq d, \abs{\beta}\leq d'}  \!\!\!\!\!\!\!\!\, \commut{\mathbf{T}_{-r}^{-1}}{\big(\prod_{j=1}^n (\mathbf{T}^*_{z_j})^{\beta_j} \big)\big(\prod_{j=1}^n \mathbf{T}_{z_j}^{\alpha_j}\big)}\\
          & = \!\!\!\!\sum_{\abs{\alpha}\leq d, \abs{\beta}\leq d'}  \!\!\!\! p_{\alpha\beta} \Bigg( \, \big(\prod_{j=1}^n (\mathbf{T}^*_{z_j})^{\beta_j} \big) \commut{\mathbf{T}_{-r}^{-1}}{\prod_{j=1}^n \mathbf{T}_{z_j}^{\alpha_j}} + \commut{\mathbf{T}_{-r}^{-1}}{\prod_{j=1}^n (\mathbf{T}^*_{z_j})^{\beta_j}}\prod_{j=1}^n \mathbf{T}_{z_j}^{\alpha_j} \, \Bigg)\\
          & = \tfrac{1}{m+1} \!\!\!\!\sum_{\abs{\alpha}\leq d, \abs{\beta}\leq d'}  \!\!\!\!\!\!\!\! p_{\alpha\beta}\, (\abs{\alpha}-\abs{\beta}) \big(\prod_{j=1}^n (\mathbf{T}^*_{z_j})^{\beta_j} \big)\big(\prod_{j=1}^n \mathbf{T}_{z_j}^{\alpha_j}\big) = \tfrac{1}{m+1}\,\!\!\!\!\sum_{\abs{\alpha}\leq d, \abs{\beta}\leq d'} \!\!\!\!\!\!\!\! p_{\alpha\beta}\, (\abs{\alpha}-\abs{\beta}) \,\mathbf{T}_{z^\alpha\bar{z}^\beta}\\
          & = \tfrac{1}{m+1}\, \mathbf{T}_{(\fat{\mathcal{R}}-\overline{\fat{\mathcal{R}}})\,p} \,.
          \tag*{\qed}
        \end{align*}
        \hideqed
      \end{proof}
        
      An example of previous integration of noncommutative spectral triples is given here for the unit ball $\B^n$:
      \begin{proposition}
        \label{propExampleSumTriple}
        For $m\in\N$, let 
        \vspace{-3mm}
        \begin{itemize}
          \item $\H_m \vc A^2_m(\mathbb{B}^n)$,
          \item $\DD_m \vc (\mathbf{T}_{-r}^{(m)})^{-1}$,
          \item $\A_m$ be the $^*$-algebra generated by Toeplitz operators $\mathbf{T}_p^{(m)}$ on $\H_m$, with $p\in\Pol(\mathbb{B}^n)$,
          \item $\pi_m$ be the identity representation on $\H_m$,
          \item $\norm{.}_m$ be the usual norm of operators,
          \item $\alpha_m \vc m+1$.
        \end{itemize}
        If we let $\H^\oplus$, $\overline{\DD^\oplus}$, $\pi^\oplus$ as in \propref{propSumST} and $\A'^\oplus$ be the algebra generated by elements of the form $(\mathbf{T}^{(m)}_{p})_{m\in\N}$, with $p\in\Pol(\mathbb{B}^n)$ (i.e. keeping the same polynomial at all levels $m\in\N$),
        then the previous quintuple is integrable and $(\A'^\oplus,\H^\oplus, \overline{\DD^\oplus})$ is a spectral triple of spectral dimension $n+1$.
      \end{proposition}

      \begin{proof}
        First, we know from \propref{propTripleBergman} that for any $m\in\N$, $(\A_m, \H_m, \DD_m)$ defines a spectral triple of dimension $n$. Moreover,
        \begin{align*}
          \norm{(1+\alpha_m^2\,\DD_m^2)^{-1/2}}_m & = \norm{(1+\alpha_m^2 \, (\mathbf{T}_{-r}^{-2})^{(m)})^{-1/2}}_m \leq \abs{\alpha_m}^{-1}\norm{\mathbf{T}_{-r}^{(m)}}_m \\
          & \leq \abs{\alpha_m}^{-1}  \norm{r}_\infty \underset{m\to +\infty}{\longrightarrow}0\,.
        \end{align*}
        Let us show that $\A'^\oplus$ is a subalgebra of $\A^\oplus$ of \propref{propSumST}: if $(a_m)_{m\in\N}=(\mathbf{T}^{(m)}_p)_{m\in\N}$ of $\A'^\oplus$, with $p\in\Pol(\mathbb{B}^n)$, is a generator, the conditions are satisfied since
        \begin{align*}
          & \sup_{m\in\N}\norm{\pi_m(a_m)}_m \leq \norm{p}_\infty < +\infty \, \,\, \text{and from \propref{propCommutatorInAlg},}\\
          & \sup_{m\in\N}\norm{\commut{\alpha_m\DD_m}{\pi_m(a_m)}}_m = \sup_{m\in\N}\tfrac{m+1}{m+1}\norm{\mathbf{T}^{(m)}_{(\fat{\mathcal{R}}-\overline{\fat{\mathcal{R}}})p}} \leq \norm{(\fat{\mathcal{R}}-\overline{\fat{\mathcal{R}}})\,p}_\infty < +\infty\,.
        \end{align*}
        These inequalities remain valid for a general element of $\A'^\oplus$, which is composed, at each level $m\in\N$, by the same finite sum of finite products of Toeplitz operators acting on $A^2_m(\B^n)$. Since $\A'^\oplus$ form a $*$-algebra, we conclude that it is a $*$-subalgebra of $\A^\oplus$ and from \propref{propSumST}, $(\A'^\oplus,\H^\oplus,\overline{\DD^\oplus})$ is a spectral triple. \\
        We now compute its spectral dimension.
        For $s\in\R$, we have 
        \begin{align*}
          \Tr \,\abs{\DD^\oplus}^{-s}  & = \sum_{m\in\N} \alpha_m^{-s} \, \Tr \,(\mathbf{T}^{(m)}_{-r})^{s} =  \sum_{m\in\N} (\tfrac{\alpha_m}{m+1})^{-s} \, \Tr \,(\fat{\mathcal{R}}+m+n+1)^{-s} \\
          & = \sum_{m\in\N}\sum_{k\in\N} \tbinom{k+n-1}{n-1} (k+m+n+1)^{-s}\,.
        \end{align*}
        For any $k\in\N$ and $s>1$, we have
        \begin{align*}
          \int_{k+m+n}^{k+m+n+1} x^{-s}dx < (k+m+n+1)^{-s} < \int_{k+m+n+1}^{k+m+n+2} x^{-s}dx\,,
        \end{align*}
        so summing over $m\in\N$ leads to
        \begin{align*}
           \tfrac{1}{s-1}(k+n)^{1-s} < \sum_{m\in\N}(k+m+n+1)^{-s} < \tfrac{1}{s-1}(k+n+1)^{1-s}\,.
        \end{align*}
	Since $\tbinom{k+n-1}{n-1} \!\!\underset{k\to+\infty}{\sim} \!\!\tfrac{k^{n-1}}{(n-1)!}$, the operator $\abs{\DD^\oplus}^{-s}$ if and only if
	\begin{align*}
	  \tfrac{1}{(1-s)(n-1)!}\sum_{k\in\N} (k+1)^{n-1} (k+n+1)^{1-s} < +\infty \quad \Leftrightarrow \quad \sum_{m\in\N} k^{n-s} < +\infty \,,
	\end{align*}
	i.e. for $s>n+1$.
      \end{proof}

      \begin{remark}
        The previous result is restricted to the case of polynomial symbols. Indeed, we cannot apply the Stone--Weierstrass theorem in order to extend the result for general smooth functions over $\overline{\B^n}$ since $f \mapsto \commut{(\mathbf{T}_{-r}^{(m)})^{-1}}{\mathbf{T}_f}$ is not continuous on $A^2_m(\B^n)$ for the norm $\norm{.}_\infty$.
      \end{remark}

      A possible extension of \propref{propExampleSumTriple}, in which any $(a_m)_{m\in\N}\in\A'^\oplus$ is defined as the copy of the same element on each level $m\in\N$, consists of replacing a finite number of $a_m$ by arbitrary elements of $\A_m$. \\
      Thus the representation of an element $a^\oplus$ of this new algebra ${\A''}^\oplus$ is of the form
      \begin{align*}
	{\pi''}^\oplus({a}^\oplus) = \bigoplus_{m\leq N} \pi_m(a_m) \oplus \bigoplus_{m>N} \,\sum_{i=1}^p\prod_{j=1}^{q_i} \mathbf{T}^{(m)}_{p_{ij}}\,,
      \end{align*}
      for some integer $N$, some arbitrary $a_m\in\A_m$, $m\leq N$, and fixed family of polynomials $p_{ij}$ in $\Pol(\mathbb{B}^n)$, $i=1,\dots,p$, $j=1,\dots,q_i$.
      
      We can also consider a more general sequence $(\alpha_m)_{m\in\N}$ such that $\alpha_m \sim m^{\delta}$, as $m$ tends to infinity, for $0< \delta \leq 1$ (the upper bound comes from the boundedness of the commutator between the representation of an element of the algebra and $\DD^\oplus$). Then, the conclusions of \propref{propSumST} remain valid but the spectral dimension lies in $[n+1,+\infty)$.
      
      The ``$n+1$ phenomenon'' also appears in the spectral dimension of the spectral triple built from the Berezin--Toeplitz quantization \cite[Section 6]{EFI2012}. For short, the latter spectral triple can be expressed as a summation of spectral triples on a smoothly bounded strictly pseudoconvex domain like the ones in \propref{propTripleBergman}. This can be equivalently seen as a spectral triple over the boundary of a disk bundle over the domain, whose spectral dimension is exactly $n+1$, and which brings a geometric explanation for the ``extra dimension''.

      Spectral triples with arbitrary real positive spectral dimension have already be encountered in C. Ivan and E. Christensen's paper \cite{CI2006}; the construction uses algebras of continuous functions over the Cantor set.

  \section{Integration along decomposable self-similar sets in the plane}
    \label{secFractalIntegration}
    
    \subsection{Spectral triple on fractal sets generated by a polygonal Jordan curve}
      \label{secIntegrationPolygonal}

      \begin{definition}
	\label{defJordanSSS}
	Let $\mathcal{S}$ be the set of families $(F_k)_{k=1,\dots,N}$ of similarities on $\C$ such that
	\begin{enumerate}[label=\itshape\roman*\upshape)]
	  \item for any $k=1,\dots,N$, $F_k$ have the same ratio $c\in (0,1)$: 
	    \begin{align*}
	      F_k(z) = a_kz +b_k\,, \quad  \quad z\in\C\,, \text{ with } \abs{a_k}=c\,,
	    \end{align*}
	  \item there is a non-empty open bounded set $V\subset \C$ such that $\bigsqcup_{k=1}^N F_k(V) \subset V$ (open set condition), 
	  \item there is a polygonal Jordan curve $E_0$ defined by the points $(p_j)_{j=1,\dots,M}$, $M>2$, such that the attractor $E$ of $(F_k)_{k=1,\dots,N}$ can be decomposed as
	  \begin{align}
	    \label{eqDecompositionE}
	    E = \overline{\bigcup_{m\in\N}\,\bigcup_{\omega\in\{1,\dots,N\}^m} \!\!\!\!\!\!\! F_{\omega_1}\circ\dots\circ F_{\omega_m}(E_0)}\,.
	  \end{align}
	\end{enumerate}
	
	The set $E_0$ is called the generator.
      \end{definition}
      The set $E$ is a non-empty closed bounded set in the metric space $\R^2$ \cite[3.1.(3)(\emph{i})]{H1981}. Moreover, since $E$ is defined from similarities of same ratios and verifies the open set condition, its Hausdorff dimension $\dim_H$ is given by (see \cite[Theorem 9.3]{F2003})
      \begin{align}
	\label{eqDimensionSSS}
	\dim_H(E)= \tfrac{\log(N)}{\log(1/c)}.
      \end{align}
    
      For the rest of this section, $(F_k)_{k=1,\dots,N}$ denotes an element of $\mathcal{S}$ with fixed ratio $c\in (0,1)$, $E_0$ a generator and $E$ the corresponding attractor. Let $(L_j)_{j=1,\dots,M-1}$ be the family of closed line segments between the points $p_j$ and $p_{j+1}$, and $L_M$ between $p_M$ and $p_1$. If $\abs{L_j}$ is the length of the segment $L_j$, we assume that the perimeter $\sum_{j=1}^M \abs{L_j}$ of $E_0$ is $2\pi$ and we denote $\theta_j \vc \sum_{\ell=1}^{j-1} \abs{L_\ell}$ for any $j=2,\dots,M$, and $\theta_1 = 0$.      
      For $j=1,\dots,M-1$, let $A_j$ be the closed arc of the unit circle $\partial\B$ between $e^{i\theta_{j}}$ and $e^{i\theta_{j+1}}$, and $A_M$ the one between the points $e^{i \theta_M}$ and 1. For any $m\in\N$, $\omega\in\{1,\dots,N\}^m$ and $j=1,\dots,M$, we use the following notations:
      \begin{equation}
	\label{eqNotations}
	\begin{aligned}
	  \B_m        & \vc \{z\in\C\,, \abs{z} < c^m\}\,, & \mathcal{C}_m  & \vc \partial\B_m \,,  & F_\omega    & \vc F_{\omega_1}\circ \dots \circ F_{\omega_m}\,,        & E_\omega & \vc F_{\omega}(E_0)\,,\\
	  p_{\omegaj} & \vc F_\omega(p_j) \,,              & L_{\omegaj}    & \vc F_\omega(L_j) \,, & A_{\omegaj} &\vc c^mA_j\,.                                                
	\end{aligned}
      \end{equation}
      (note that $A_{\omegaj}$ are the closed arcs whose union over $j$ is $\mathcal{C}_m$.)
      
     \begin{example}
      The Sierpinski gasket $E_{\mathcal{SG}}$ \cite{S1915} is the attractor of $(F_1,F_2,F_3) \in\mathcal{S}$, where $F_k$, $k=1,2,3$, is the homothety of center the $k$th vertex $p_k$ of an equilateral triangle $E_0$, and of ratio $c=1/2$. It can be expressed as the union between $E_0$ and all of its images by $F_\omega$, for any $\omega\in\{1,\dots,N\}^m$ and $m\in\N$, and its Hausdorff dimension is $\dim_H(E_{\mathcal{SG}}) = \tfrac{\log(3)}{\log(2)}$. On \figref{figSierpinskiGasket}, the sets $F_k(E_0)$ and $F_k\circ F_l(E_0)$, $k,l=1,2,3$, are denoted $E_k$ and $E_{kl}$ respectively. 
	\begin{figure}[H]
	  \centering
	  \begin{tikzpicture}[scale=0.8]
	    \draw[dashed,gray] (0,0) -- ++(60:3) -- ++(-60:3) -- cycle;
	    
	    \draw (1.5,0) -- ++(60:1.5) -- ++(180:1.5) -- cycle;
	    
	    \draw (0,0) node (a) {};
	    \draw (3,0) node (b) {};
	    \draw (1.5,2.5981) node (c) {};
	    
	    \draw (a) node[left] {~~$p_1$};
	    \draw (b) node[right] {$p_2$};
	    \draw (c) node[above] {$p_3$};
	    
	    \draw (0.3,2.5981/2) node {$E_0$};
	    
	    \draw (1.5,-0.7) node[color=black] {\footnotesize Step 0};
	  \end{tikzpicture}
	  \begin{tikzpicture}[decoration=recsierpinski2,scale=0.8]
	    \draw[dashed,gray] (0,0) -- ++(60:3) -- ++(-60:3) -- cycle;
	  
	    \draw decorate { (1.5,0) -- ++(60:1.5) -- ++(180:1.5) -- cycle };
	    
	    \draw (-0.25,0.74) node {\small$E_1$};
	    \draw (3.25,0.74) node {\small$E_2$};
	    \draw (0.5,1.94) node {\small$E_3$};
	    \draw (a) node[left] {~~{\color{white}$p_1$}};
	    \draw (b) node[right] {{\color{white}$p_2$}};
	    \draw (c) node[above] {{\color{white}$p_3$}};
	    
	    \draw (1.5,-0.7) node[color=black] {\footnotesize Step 1};
	  \end{tikzpicture}
	  \begin{tikzpicture}[decoration=recsierpinski2,scale=0.8]
	    \draw[dashed,gray] (0,0) -- ++(60:3) -- ++(-60:3) -- cycle;
	  
	    \draw decorate { decorate { (1.5,0) -- ++(60:1.5) -- ++(180:1.5) -- cycle} };
	    
	    \draw (-0.1,2.5981/5 - 0.1) node {\footnotesize$E_{11}$};
	    \draw (0.2,2.5981*2/5 - 0.1) node {\footnotesize$E_{13}$};
	    \draw (2.5981*2/5+0.1,-0.2) node {\footnotesize$E_{23}$};
	    
	    \draw (2.5981 * 3/5 +0.3, -0.2) node {\footnotesize$E_{21}$};
	    \draw (3+0.1,2.5981/5 - 0.1) node {\footnotesize$E_{22}$};
	    \draw (3-0.2,2.5981*2/5 - 0.1) node {\footnotesize$E_{23}$};

	    \draw (0.6,2.5981*4/5 - 0.4) node {\footnotesize$E_{31}$};
	    \draw (2.4,2.5981*4/5 - 0.4) node {\footnotesize$E_{21}$};
	    \draw (1,2.5981*4/5 + 0.3) node {\footnotesize$E_{33}$};
	    
	    \draw (a) node[left] {~~{\color{white}$p_1$}};
	    \draw (b) node[right] {{\color{white}$p_2$}};
	    \draw (c) node[above] {{\color{white}$p_3$}};
	    
	    \draw (1.5,-0.7) node[color=black] {\footnotesize Step 2};
	  \end{tikzpicture}
	  \begin{tikzpicture}[decoration=recsierpinski2,scale=0.8]
	    \draw[black] (0,0) -- ++(60:3) -- ++(-60:3) -- cycle;
	  
	    \draw decorate { decorate { decorate { decorate { (1.5,0) -- ++(60:1.5) -- ++(180:1.5) -- cycle} } } };
	    \draw (0.5,1.94) node {\small$E$};
	    
	    \draw (a) node[left] {~~{\color{white}$p_1$}};
	    \draw (b) node[right] {{\color{white}$p_2$}};
	    \draw (c) node[above] {{\color{white}$p_3$}};
	    
	    \draw (1.5,-0.7) node[color=black] {\color{white}\footnotesize Step 2};
	  \end{tikzpicture}
	  \caption{}
	  \label{figSierpinskiGasket}
	\end{figure}
    \end{example}
    
    In order to define Toeplitz operators on the polygonal Jordan curves $E_\omega$, we define a sufficiently nice homeomorphism from the circle $\mathcal{C}_m$ into $E_\omega$ based on M\"obius transforms that send holomorphically each arc $A_{\omegaj}$ into the line segments $[p_\omegaj,p_{\omega j+1}]$.
    
    \begin{lemma}
      \label{lemDefKappa}
      For any $m\in\N$, $\omega\in\{1,\dots,N\}^m$ and $j=1,\dots,M$, let the maps
      \begin{align*}
	\kappa_{\omegaj} (z) \vc \tfrac{(p_\omegaj - i/(\delta_\omegaj \tau_\omegaj))z \, + \, c^me^{i\theta_{j}}(p_\omegaj + i/(\delta_\omegaj \tau_\omegaj))     }{ z \,+\, c^m e^{i\theta_{j}}}\,,
      \end{align*}
      from $\mathcal{C}_m$ into $E_\omega$, where
      \begin{align*}
        &\delta_{\omegaj} \vc (p_{\omega j+1} -p_{\omegaj})^{-1}\,, \quad \tau_\omegaj \vc \tan((\theta_{j+1}-\theta_{j})/2)\,,\quad \text{ if }j=1,\dots,M-1\,, \text{ and }\\
        & \delta_{\omega M} \vc (p_{\omega 1} -p_{\omega M})^{-1}\,, \quad \tau_{\omega M} \vc \tan((\theta_{1}-\theta_{M})/2)\,.
      \end{align*}
      Then the map $\kappa_\omega$ defined as $\kappa_{\omega} \evalat{A_\omegaj} \vc \kappa_{\omegaj}$ is an homeomorphism from $\mathcal{C}_m$ into $E_\omega$.
    \end{lemma}
    \begin{proof}
      Since $M>2$, $\abs{L_j}=\theta_{j+1}-\theta_j < \pi$ so $\tau_{\omegaj} < +\infty$ for any $j=1,\dots,M$.
      Expressing any point $z$ in $A_{\omegaj}$ as 
      \begin{align}
	\label{eqzInArc}
        z=z(t) = c^m e^{i(\theta_j + t\abs{L_j})} \,, \quad \text{ for }t\in [0,1]\,,
      \end{align}
      we have for any $t\in [0,1]$
      \begin{align}
        \kappa_\omegaj(z(t)) & = \tfrac{(p_\omegaj - i/(\delta_\omegaj \tau_\omegaj))c^m e^{i(\theta_j + t\abs{L_j})} \, + \, c^me^{i\theta_{j}}(p_\omegaj + i/(\delta_\omegaj \tau_\omegaj))     }{ c^m e^{i(\theta_j + t\abs{L_j})} \,+\, c^m e^{i\theta_{j}}} \nonumber \\
        & = \tfrac{ (p_\omegaj - i/(\delta_\omegaj\tau_\omegaj))e^{it\abs{L_j}} + p + i/(\delta_\omegaj \tau_\omegaj)   }{ e^{it\abs{L_j}} + 1 } = p_\omegaj + \tfrac{i}{\delta_\omegaj\tau_\omegaj} \tfrac{1-e^{it\abs{L_j}}}{1+e^{it\abs{L_j}}} \nonumber \\
        \label{eqExprKappa}
        & = p_\omegaj + (p_{\omega j+1} - p_\omegaj) \tfrac{\tan( t \,\abs{L_j}/2  )}{\tan(\abs{L_j}/2)}\,.
      \end{align}
      Thus $\kappa_\omegaj$ maps continuously the closed arc $A_{\omegaj}$ into the line segment $[p_\omegaj,p_{\omega j+1}]$ for any $j=1,\dots,M-1$, so does $\kappa_M$ from $A_M$ into the line segment $[p_{\omega M}, p_{\omega 1}]$.
    \end{proof}
    
    Let $m\in\N$ and $\omega\in\{1,\dots,N\}^m$. The \emph{Hardy space} over the circle $\mathcal{C}_m$, denoted $H^2(\mathcal{C}_m)$, $\H_\omega$ or $\H_m$, is the space of functions $\phi$ that are holomorphic on the corresponding open disk $\B_m$ and such that 
    \begin{align*}
      \norm{\phi}_{\H_m}^2 \vc \sup_{0<\rho< c^m} \int_{0}^{2\pi} \abs{\phi(\rho e^{it})}^2 \, \tfrac{dt}{2\pi} < +\infty\,.
    \end{align*}
    The inner product in $\H_m$ between $\phi(z) = \sum_{k\in\N} \phi_k z^k$ and $\psi(z) = \sum_{k\in\N} \psi_k z^k$, $z\in\B_m$, is given by
    \begin{align*}
      \scalp{\phi}{\psi}_{\H_m} \vc \int_{0}^{2\pi} \phi(c^me^{it}) \, \overline{\psi(c^me^{it})} \, \tfrac{dt}{2\pi} = \sum_{k\in\N} \phi_k\overline{\psi}_k\, c^{2m k}\,.
    \end{align*}
    An orthonormal basis for $\H_m$ is given by the vectors $v_{mj}(z) \vc c^{-m}z^j$ and the reproducing kernel $S^{(m)}_z(w) = \sum_{k\in\N} c^{-2mk}\bar{z}^kw^k$ verifies $\phi(z) = \scalp{\phi}{S^{(m)}_z}_{\H_m}$ for any $\phi\in \H_m$ and $z\in\B_m$. The orthogonal projection $\Pi_m : L^2(\mathcal{C}_m) \to \H_m$ is called the Szeg\"o projector. Any bounded function $u$ on the circle $\mathcal{C}_m$ gives rise to the Toeplitz operator $T_u : \phi \mapsto \Pi_m(u\phi)$, $\phi\in \H_m$. Thus for any polynomial function $p \in \Pol(E_\omega)$, the function $p\circ \kappa_\omega$ is bounded on $\mathcal{C}_m$, and we can consider Toeplitz operators of the form
    \begin{align}
      \label{eqTomegaPol}
      T^{(\omega)}_{p\circ \kappa_\omega} : \H_m \ni \phi \mapsto \Pi_m \big((p\circ \kappa_\omega)\phi \big) \in \H_m\,.
    \end{align}
    The integral representation of the action of such operators is
    \begin{equation}
      \begin{aligned}
	\label{eqIntRepToeplitzHardy}
	\big(T^{(\omega)}_{p\circ \kappa_\omega} \phi \big)(z) & = \scalp{(p\circ \kappa_\omega) \phi}{S^{(m)}_z}_{\H_\omega} \\
	& = \sum_{j=1}^M \int_{\theta_j}^{\theta_{j+1}}\!\!\!\tfrac{dt}{2\pi} (p\circ \kappa_{\omegaj})(c^me^{it})\, \phi(c^me^{it}) \, \overline{S_{z}^{(m)}}(c^me^{it}) \,,\quad z\in \B_m\,.
      \end{aligned}
    \end{equation}
    As in previous section, define the operators  $\mathcal{R} \vc z\partial_z$ and $\overline{\mathcal{R}} \vc \bar{z}\partial_{\bar{z}}$ acting on $\H_\omega$.\\ Since $\Spec(\mathcal{R})=\N$ consists of the positive part of the usual Dirac operator on the circle $\mathcal{C}_m$, we choose for $\DD_\omega$ the following expression 
    \begin{align}
      \label{eqDm}
      \DD_\omega \vc \alpha_m \mathcal{R} + \beta_m\,,
    \end{align}
    where $(\alpha_m)_{m\in\N}$ and $(\beta_m)_{m\in\N}$ are two sequences of strictly positive real numbers: these sequences depend only on $m$ since each $\DD_\omega$ acts on a Hardy space over the same circle $\mathcal{C}_m$ of radius $c^{m}$. 
    
    \begin{lemma}
      \label{lemCommutatorBounded}
      Let $m\in\N$, $\omega\in\{1,\dots,N\}^m$ and $p\in \Pol(E_w)$. The operator $\commut{\DD_\omega}{T_{p\circ\kappa_\omega \evalat{\mathcal{C}_\omega}}^{(\omega)}}$ is bounded and 
      \begin{align*}
        \norm{\commut{\DD_\omega}{T_{p\circ\kappa_\omega \evalat{\mathcal{C}_\omega}}^{(\omega)}}} \leq \alpha_m \,K_p\,,
      \end{align*}
      for some constant $K_p>0$ independent on $m$.
    \end{lemma}
    \begin{proof}
      We fix $m\in\N$, $\omega\in\{1,\dots,N\}^m$, and denote $T_u \vc T_u^{(w)}$, $u\in L^\infty(\mathcal{C}_\omega)$, for clarity reasons. If $p(z) = \sum_{a,b} p_{ab} \, z^k \,\bar{z}^l$ is a polynomial on the set $E_\omega$, we have $\commut{\mathcal{R}}{T_{p\circ \kappa_\omega}} = \sum_{a,b} p_{ab} \, \commut{\mathcal{R}}{T_{\kappa_\omega^a\overline{\kappa_\omega}^b}}$. We now show that for any $a,b\in\N$, 
      \begin{align}
	\label{eqEqualityCommutpiab}
        \commut{\mathcal{R}}{T_{\kappa_\omega^a\overline{\kappa_\omega}^b}} = T_{\pi_{\omega ab}}\,,
      \end{align}
      as operators acting on $\H_\omega$, where $\pi_{\omega ab}$ defined by
      \begin{align*}
	\pi_{\omega ab}\evalat{\mathring{A}_\omegaj} = (\mathcal{R}-\overline{\mathcal{R}})\kappa^a_\omegaj \overline{\kappa_\omegaj}^b\,,  \quad \text{ for any } j=1,\dots,M\,,
      \end{align*}
      where $\mathring{A}_\omegaj$ denotes the interior of the closed arc $A_\omegaj$, is extendible to a piecewise continuous, hence bounded, function on $\mathcal{C}_m$. For any $j=1,\dots, M$, the map $\kappa_{\omegaj}^a$ is holomorphic in a neighborhood of $A_\omegaj$, and we write $\kappa_{\omegaj}^a (z) = \sum_{k\in\N} \kappa_{\omegaj ak} \,z^k$, $z\in A_\omegaj$. For any $z\in\B_m$, $n\in\N$, and using \eqref{eqIntRepToeplitzHardy}, we have on one hand
      \begin{align*}
        T_{\kappa^a_\omega \overline{\kappa_\omega}^b} z^n & = \sum_{j=1}^M\int_{\theta_j}^{\theta_{j+1}} \!\!\!\tfrac{dt}{2\pi} \sum_{k\in\N} \kappa_{\omegaj ak}\,c^{mk}e^{itk}\, \overline{ \sum_{l\in\N} \kappa_{\omegaj bl} \,c^{ml}e^{itl}} \, c^{mn}e^{itn}\,\sum_{s\in\N} c^{-2ms}z^s\,c^{ms}e^{-its}\\
        & = \sum_{s\in\N}c^{m(n-s)}z^s \, \sum_{j=1}^M \int_{\theta_j}^{\theta_{j+1}} \!\!\!\tfrac{dt}{2\pi} \sum_{k,l\in\N} \kappa_{\omegaj ak}\overline{\kappa_{\omegaj bl}} \, c^{m(k+l)}\, e^{it(k-l+n-s)} \\
        & = \sum_{s\in\N}c^{m(n-s)}z^s  \sum_{j=1}^M  \sum_{k,l\in\N} \!\!\kappa_{\omegaj ak}\overline{\kappa_{\omegaj bl}}\, c^{m(k+l)} \\ 
        & \hspace{1cm} \Big( \tfrac{e^{i\theta_{j+1}(k-l+n-s)} - e^{i\theta_{j}(k-l+n-s)}}{i(k-l+n-s)} \, \delta(k+n\neq l+s)  + (\theta_{j+1}-\theta_j)\,\delta(k+n = l+s) \, \Big)\,.
      \end{align*}
      On the other hand, since on each $A_\omegaj$ we have $(\mathcal{R}-\overline{\mathcal{R}})\kappa_{\omegaj}^a\overline{\kappa_\omegaj}^b = (\mathcal{R}\kappa_\omegaj^a)\,\overline{\kappa_\omegaj}^b - \kappa_\omegaj^a \overline{\mathcal{R}}\,\overline{\kappa_\omegaj}^b$, we get for any $z\in\B_m$ and $n\in\N$
      \begin{align*}
        T_{\pi_{\omega ab}} z^n & = \sum_{j=1}^M\int_{\theta_j}^{\theta_{j+1}} \!\!\!\tfrac{dt}{2\pi} \sum_{k,l\in\N} (k-l)\,\kappa_{\omegaj ak}\overline{\kappa_{\omegaj bl}}\,c^{m(k+l)}\, c^{mn}e^{itn}\,\sum_{s\in\N} c^{-2ms}z^s\,c^{ms}e^{-its}\\
        & = \sum_{s\in\N}c^{m(n-s)}z^s  \sum_{j=1}^M  \sum_{k,l\in\N} \!\!\kappa_{\omegaj ak}\overline{\kappa_{\omegaj bl}}\, c^{m(k+l)} (k-l)\\ 
        & \hspace{1cm} \Big( \tfrac{e^{i\theta_{j+1}(k-l+n-s)} - e^{i\theta_{j}(k-l+n-s)}}{i(k-l+n-s)} \, \delta(k+n\neq l+s)  + (\theta_{j+1}-\theta_j)\,\delta(k+n = l+s) \, \Big)\,.
      \end{align*}
      Thus, setting $B \vc T_{\pi_{\omega ab}}  \!\!- \!\commut{\mathcal{R}}{T_{\kappa^a_\omega \overline{\kappa_\omega}^b}}$, we obtain
      \begin{align*}
         Bz^n & =  \sum_{s\in\N}c^{m(n-s)}z^s  \sum_{j=1}^M  \sum_{k,l\in\N} \!\!\kappa_{\omegaj ak}\overline{\kappa_{\omegaj bl}}\, c^{m(k+l)} (k-l+n-s)\\ 
        & \hspace{4cm}  \tfrac{e^{i\theta_{j+1}(k-l+n-s)} - e^{i\theta_{j}(k-l+n-s)}}{i(k-l+n-s)} \, \delta(k+n\neq l+s)\\
        & = -i\sum_{s\in\N}c^{m(n-s)}z^s  \sum_{j=1}^M \sum_{k,l\in\N} \!\!\kappa_{\omegaj ak}\overline{\kappa_{\omegaj bl}}\, c^{m(k+l)} (e^{i\theta_{j+1}(k-l+n-s)} - e^{i\theta_{j}(k-l+n-s)}) \\
        & = -i\sum_{s\in\N}c^{m(n-s)}z^s e^{i}  \sum_{j=1}^M \big(\kappa_{\omegaj}^a(p_{\omega j+1})\overline{\kappa_\omegaj}^b(p_{\omega j+1})e^{i\theta_{j+1}(n-s)} -\kappa^a_{\omegaj}(p_{\omegaj})\overline{\kappa_\omegaj}^b(p_{\omegaj})e^{i\theta_{j}(n-s)} \big)\,.
      \end{align*}
      Since $\kappa_\omega$ is continuous on $\mathcal{C}_m$, $\kappa_\omegaj(p_{\omega j+1}) = \kappa_{\omega j+1}(p_{\omega j+1})$ for any $j=1,\dots,M-1$, and $\kappa_{\omega M}(p_{\omega 1}) = \kappa_{\omega 1}(p_{\omega 1})$, so the summation over $j$ on the right-hand side vanishes and we proved \eqref{eqEqualityCommutpiab}.\\
      By linearity, we have $\commut{\DD_\omega}{T_{p\circ \kappa_\omega\evalat{\mathcal{C}_m}}} = \alpha_m\,\sum_{a,b}p_{ab} \, T_{\pi_{\omega ab}}$, which is a bounded operator on $\H_\omega$ with
      \begin{align*}
        \norm{\commut{\DD_\omega}{T_{p\circ \kappa_\omega\evalat{\mathcal{C}_m}}}} \leq \alpha_m \sum_{a,b}\abs{p_{ab}} \norm{\pi_{\omega ab}}_\infty\,.
      \end{align*}
      We have for any $z\in \mathring A_\omegaj$
      \begin{align*}
	\mathcal{R}\kappa_{\omegaj}(z) = z\partial_z\tfrac{(p_\omegaj - i/(\delta_\omegaj \tau_\omegaj))z \, + \, c^me^{i\theta_{j}}(p_\omegaj + i/(\delta_\omegaj \tau_\omegaj))     }{ z \,+\, c^m e^{i\theta_{j}}} = \tfrac{-2i/(\delta_\omegaj\tau_\omegaj)c^me^{i\theta_j}z}{(z+c^me^{i\theta_j})^2}\,,
      \end{align*}
      so using \eqref{eqzInArc} we get
      \begin{align*}
        \mathcal{R}\kappa_{\omegaj}(t) = \tfrac{-2i}{\delta_{\omegaj} \tau_{\omegaj}} \tfrac{e^{it \,\abs{L_j}}}{(1+e^{it \,\abs{L_j}})^2}\,, \quad \text{ hence } \quad \abs{\mathcal{R}\kappa_{\omegaj}}(t) \leq \tfrac{2}{\tau_M \abs{p_{j+1}-p_j} k_j} \,, \quad \forall t\in (0,1)\,,
      \end{align*}
      with $k_j \vc \inf_{t\in (0,1)} (1+e^{it\, \abs{L_j}})^2 > 0$. Moreover, since $E$ is a compact set in $\C$, there is a constant $K>0$ such that for any $m\in\N$ and $\omega\in \{1,\dots,N\}^m$, $\norm{\kappa_{\omega}}_\infty \leq K$. Thus       \begin{align*}
        \norm{\pi_{\omega ab}}_\infty & \leq \sup_{j=1,\dots,M} \norm{ a(\mathcal{R}\kappa_\omegaj)\kappa_\omegaj^{a-1}\overline{\kappa_\omegaj}^b \evalat{\mathring A_{\omegaj}}}_\infty +\norm{ b\kappa_\omegaj^a(\overline{\mathcal{R}}\overline{\kappa_\omegaj})\overline{\kappa_\omegaj}^{b-1} \evalat{\mathring A_{\omegaj}}}_\infty\\
        & \leq \sup_{j=1,\dots,M} \tfrac{2(a+b) K^{a+b-1}}{\tau_{\omegaj}\abs{p_{\omegaj +1} - p_\omegaj} k_j} \cv K'_{ab}\,. 
      \end{align*}
      Finally, we take $K_p \vc \sum_{a,b}\abs{p_{ab}} K'_{ab}$ and the proof is complete.
    \end{proof}
    
    \begin{proposition}
      \label{propSTAwHwDw}
      Let $\omega \in \{1,\dots,N\}^m$, $m\in\N$. Let $\A_\omega$ be the algebra generated by Toeplitz operators of the form \eqref{eqTomegaPol}, $\H_\omega \vc H^2(\mathcal{C}_m)$ and $\DD_\omega$ as in \eqref{eqDm}.\\
      Then for any sequences $(\alpha_m)_{m\in\N}$, $(\beta_m)_{m\in\N}$ of stricly positive real numbers, $(A_\omega, \H_\omega, \DD_\omega)$ is a spectral triple of spectral dimension $1$.
    \end{proposition}
    \vspace{-5mm}
    \begin{proof}
      The circle $\mathcal{C}_\omega$ is the boundary of a strictly pseudoconvex domain with complex dimension 1 and $\A_\omega$ is a subalgebra of the algebra of generalized Toeplitz operators of order 0 (see \cite{BdMG1981}), and the operator $\DD_\omega$ is a selfadjoint elliptic generalized Toeplitz operator of order 1 on $\H_\omega$, so the proof is similar to the one of \cite[Proposition 5.2]{EFI2012}, except that here $\commut{\DD_\omega}{\A_\omega} \notin \A_\omega$. The boundedness of the commutator is nonetheless proved by \lemref{lemCommutatorBounded}.
    \end{proof}
    
    We assumed that the attractor $E$ is the union of all the components $E_\omega$, so we sum the spectral triples obtained in \propref{propSTAwHwDw} in order to obtain a noncommutative spectral triple over the whole set $E$. It is still possible to adjust the coefficients $(\alpha_m)_{m\in\N}$ and $(\beta_m)_{m\in\N}$ so that the spectral dimension of the integrated spectral triple corresponds to $\dim_H(E)$.
    
    \begin{theorem}
      \label{theoSumSTIFS}
      Let $(F_k)_{k=1,\dots,N}$ be an element of $\mathcal{S}$ with ratio $c\in(0,1)$ such that $1<cN$. Let
      \begin{itemize}
        \item $\H^\oplus \vc \bigoplus_{m\in\N} \bigoplus_{\omega \in \{1,\dots,N\}^m} \H_\omega$, with $\H_\omega \vc H^2(\mathcal{C}_m)$,
        \item $\A^\oplus$ be the algebra generated by Toeplitz operators of the form 
        \begin{align*}
          T^\oplus_p \vc \bigoplus_{m\in\N} \bigoplus_{\omega \in \{1,\dots,N\}^m} \!\!\!\!T^{(\omega)}_{p \circ \kappa_\omega}\,, \quad \text{ with $p \in \Pol(E)$,}
        \end{align*}
        \item $\DD^\oplus \vc \bigoplus_{m\in\N} \bigoplus_{\omega \in \{1,\dots,N\}^m}  \DD_\omega$, where $\DD_\omega \vc \alpha_m\mathcal{R}_\omega + \beta_m$, for some sequences $(\alpha_m)_{m\in\N}$, $(\beta_m)_{m\in\N}$ of strictly positive real numbers and with domain
        \begin{align*}
	  \dom(\DD^\oplus) \vc \{v^\oplus \vc \bigoplus_{m=0}^N\bigoplus_{\omega \in  \{1,\dots,N\}^m} \!\!\!\!\!\!\!\! v_\omega \,, N\in\N\,, v_\omega \in \dom(\DD_\omega)\}\,.
        \end{align*}
      \end{itemize}
      Then $\DD^\oplus$ is essentially selfadjoint and one can choose the sequences $(\alpha_m)_{m\in\N}$ and $(\beta_m)_{m\in\N}$ so that $(\A^\oplus,\H^\oplus,\overline{\DD^\oplus})$ is a spectral triple of spectral dimension $\dim_H(E)$. 
     \end{theorem}

    \begin{proof}
      The attractor $E$ is compact, so for any $p\in\Pol(E)$, the norm $\norm{T_p^{\oplus}} \leq \norm{p}_\infty$ is finite.\\
      Let 
      \begin{align}
        \label{eqConditionEll}
        \ell\in (\tfrac{\log(N)}{\log(cN)}, +\infty)  \,, \quad  \alpha_m \vc c^{-\ell m}N^{-m(\ell-1)}\,, \quad \text{ and } \quad \beta_m \vc c^{-\ell m}\,, \quad \forall m\in\N\,.
      \end{align}
      For any $m\in\N$ and $\omega \in \{1,\dots,N\}^m$, $\Spec(\DD_\omega)=\{\alpha_m j + \beta_m\,, j\in\N\} \subset \R^{+}\noNull$, so the operator $\DD_\omega^{-1}$ is compact and 
      \begin{align*}
        \norm{\DD_\omega^{-1}} = \sup_{j\in\N} (\alpha_m j + \beta_m)^{-1} = \beta_m^{-1} = c^{\ell m} \underset{m\to+\infty}{\longrightarrow} 0\,,
      \end{align*}
      hence $\DD^\oplus$ has compact resolvent. Using \lemref{lemCommutatorBounded}, for any $p\in\Pol(E)$, we have
      \begin{align*}
        \norm{\commut{\DD^\oplus}{T_p^\oplus}} = \sup_{m\in\N}\sup_{\omega\in\{1,\dots,N\}^m} \norm{\commut{\DD_\omega}{T_{p\circ\kappa_\omega}^{(\omega)}}} \leq K_p\sup_{m\in\N}\alpha_m \leq K_p
      \end{align*}
      (indeed $\alpha_0 = 1$ and $\tfrac{\log(N)}{\log(cN)} <\ell \Leftrightarrow \alpha_m < 1$ for any $m\neq 0$).\\
      From \propref{propSTAwHwDw}, the spectral dimension of $(\A_\omega, \H_\omega, \DD_\omega)$ is 1 for any $\omega\in\{1,\dots,N\}^m$, $m\in\N$, so we study $\Tr(\abs{\DD^{\oplus}}^{-s})$ for $s>1$:
      \begin{align*}
        \Tr(\abs{\DD^{\oplus}}^{-s}) = \sum_{m\in\N} N^m \sum_{j\in\N}(\alpha_m j +\beta_m)^{-s} = \sum_{m\in\N} N^m \alpha_m^{-s}\sum_{j\in\N}(j +\tfrac{\beta_m}{\alpha_m})^{-s}\,.
      \end{align*}
      A similar calculation as in the proof of \propref{propExampleSumTriple} shows that $\sum_{j\in\N}(j +\tfrac{\beta_m}{\alpha_m})^{-s}\!\!\!\!\underset{m\to+\infty}{\sim} \!\!\tfrac{1}{s-1}(\tfrac{\beta_m}{\alpha_m})^{1-s}$, so $\Tr(\abs{\DD^{\oplus}}^{-s})$ is finite if and only if 
      \begin{align*}
	\sum_{m\in\N} N^m\alpha_m^{-s} (\tfrac{\beta_m}{\alpha_m})^{1-s} = \sum_{m\in\N} (c^{\ell s} N^\ell )^m < + \infty\, \quad \Leftrightarrow \quad c^{\ell s}N^{\ell} < 1\,,
      \end{align*}
      i.e. for $s>\tfrac{\log(N)}{\log(1/c)} = \dim_H(E)$.
    \end{proof}
    Because of the condition $1<cN$, the previous operator $\DD^\oplus$ encodes the Hausdorff dimension of $E$ when the latter is strictly greater than 1. 

    \subsection{Spectral triple on fractal sets generated by the unit disk}
      
      The integration of \secref{secExampleIntegration} is obtained from a family of spectral triples over a fixed domain which is the unit disk $\B$, and the dimension is recovered after adjusting the family of weights and the sequence $(\alpha_m)_{m\in\N}$. Here, the considered domain is a union of disks of different sizes, forming a self-similar set. This approach of integration seems more natural in the sense that the geometrical structure of the fractal keeps us to put by hand the coefficients $\alpha_m$ on the operators $\DD_\omega$, and the only degree of freedom remains the choice of the weights on each disk. 
      
      In this section, we consider a family of similarities $(F_k)_{k=1,\dots,N}$ on $\C$ which verify \emph{i)} and \emph{ii)} from \defref{defJordanSSS}, whose attractor $E$ can be expressed as
      \begin{align*}
	E = \bigcup_{m\in\N}\,\bigcup_{\omega\in\{1,\dots,N\}^m} \!\!\!\!\!\!\! F_\omega(\B)\,
      \end{align*}
      (we keep the same notations as in \eqref{eqNotations}).
      
      Again, the attractor $E$ is a self-similar set and its Hausdorff dimension is given by \eqref{eqDimensionSSS}. The spectral triple over $E$ is obtained in a similar way as in \secref{secIntegrationPolygonal}, except that the algebras we consider here are the algebras generated by Toeplitz operators on Bergman spaces over the disks.

      For any $m\in\N$, we consider for the disk $\B_m$ the following defining function and weight 
      \begin{align*}
        r_{m}(z) \vc  \abs{z}^2-c^{2m}\,, \quad \text{ and } \quad \widetilde{w}_{m}(z) \vc (-r_m(z))^{N^m}\,, \quad z\in \overline{\B_m}\,,
      \end{align*}
      and we denote the corresponding weighted Bergman spaces $\widetilde{A}^2_m(\B_m)$. For any $m\in\N$ and $\omega\in\{1,\dots,N\}^m$, we consider the translation $\iota_\omega : z \mapsto z + q_\omega$ from $\B_m$ into $F_\omega(\B)$, where $q_\omega \vc F_\omega(0)$ is the center of the open disk $F_\omega(\B)$. Of course \propref{propTripleBergman} remains valid when $\H_m$ is replaced by $\widetilde{A}^2_m(\B_m)$, and we get
      
      \begin{proposition}
        Let 
        \begin{itemize}
          \item $\H^{\oplus} \vc \bigoplus_{m\in\N} \bigoplus_{\omega\in\{1,\dots,N\}}^{m} \H_m$, with $\H_m \vc \widetilde{A}^2_m(\B_m)$,
          \item $\A^{\oplus}$ be the $*$-algebra generated by operators on $\H^\oplus$ of the form 
          \begin{align*}
            \fatT_p^{\oplus} \vc \bigoplus_{m\in\N}\bigoplus_{\omega\in\{1,\dots,N\}^{m}}\fatT_{p\circ \iota_\omega \evalat{\B_m}}^{(\omega)}\,, \quad \text{ with } p\in\Pol(E)\,,
          \end{align*}
          \item $\DD^{\oplus} \vc \bigoplus_{m\in\N}\bigoplus_{\omega\in\{1,\dots,N\}^{m}}\DD_\omega$, with $\DD_\omega \vc (\fatT_{-r_{m}}^{(\omega)})^{-1}$.
        \end{itemize}
        Then, if $1<c^2N$, then $(\A^\oplus, \H^{\oplus}, \DD^{\oplus})$ is a spectral triple of spectral dimension $\dim_{H}(E)$.
      \end{proposition}
      
      \begin{proof}
	For any $m\in\N$, an orthonormal basis of $\widetilde{A}^2_m(\B_m)$ is given by
	\begin{align*}
	  \widetilde{u}_{m,j} (z) \vc c^{-m(N^m+j+1)} \big( \tfrac{(N^m+j+1)!}{N^m! \, j! \, \pi }\big)^{1/2}\, z^j\,, \quad z\in\B_{m}\,,
	\end{align*}
	and similarly as in \eqref{eqRelationsShift}, we have for any $\omega\in\{1,\dots,N\}^m$
	\begin{align*}
          \fatT_{z_j}^{(\omega)} &= c^m \fat{S}\big( \tfrac{\fatR+1}{\fatR+N^m+2} \big)^{1/2}\,, \quad \text{ and } \quad \DD_\omega = c^{-2m}\, \tfrac{\fatR + N^m +2}{N^m + 1} = \alpha'_m \fatR + \beta_m'\,,
        \end{align*}
        with $\alpha_m' \vc c^{-2m}(N^m-1)^{-1}$ and $\beta_m \vc c^{-2m}\tfrac{N^m+2}{N^m+1}$. These sequences are equivalent, as $m\to\infty$, to $(\alpha_m)_{m\in\N}$ and $(\beta_m)_{m\in\N}$ of \eqref{eqConditionEll} when $\ell = 2$. Since we assumed $1<c^2N$, the compactness of the resolvent of $\DD^\oplus$ and the computation of the spectral dimension are shown similarly as in the proof of \theoref{theoSumSTIFS}. Moreover, as in \propref{propCommutatorInAlg}, we have $\commut{\DD_\omega}{\fatT_{p\circ \iota_\omega}^{(\omega)}} = \alpha'_m \fatT_{(\fatR - \overline{\fatR})(p\circ \iota_\omega)}^{(\omega)}$ for any $m\in\N$ and $\omega\in\{1,\dots,N\}^m$, so $\norm{\commut{\DD_\oplus}{\fatT^\oplus_p}} \leq \sup_{m\in\N} \alpha_m' \norm{p}_\infty < +\infty$.
      \end{proof}
  
  \textbf{Acknowledgements:} The author expresses his gratitude to M. Engli\v s, B. Iochum and O. Gabriel for useful advices.
  
  
  \phantomsection
  \bibliographystyle{apalike}

\end{document}